\documentclass[submission,copyright,creativecommons]{eptcs}
\providecommand{\event}{QPL2017 Proceedings} 
\usepackage{breakurl}             
\usepackage{underscore}           
\usepackage{amssymb,latexsym,amsmath,graphicx}
\usepackage[utf8]{inputenc}

\usepackage{mathabx}

\usepackage{tikz,ifthen}
\usetikzlibrary{calc}

\usepackage{stmaryrd}
\usepackage{url}
\usepackage{hyperref}

\usepackage{multicol}

\newtheorem{theorem}{Theorem}
\newtheorem{lemma}[theorem]{Lemma}
\newtheorem{corollary}[theorem]{Corollary}

\newtheorem{observation}[theorem]{Observation}

\newtheorem{definition}[theorem]{Definition}

\newcommand{\qed}{\hspace{\stretch{1}} $\Box$}
\newenvironment{proof*}[1]{\noindent{\bf #1}}{\qed}
\newenvironment{proof}{\noindent {\bf Proof:}}{\qed}

\newcommand{\inline}[1]{{\color{red!70!black} \tt [\,#1\,]}}
\newcommand{\attention}[1]{\begin{quote}\inline{#1}\end{quote}}

\newcommand{\delete}[1]{}

\newcommand{\red}[1]{{\color{red!70!black}{#1}}}
\newcommand{\green}[1]{{\color{green!70!black}{#1}}}
\newcommand{\blue}[1]{{\color{blue!70!black}{#1}}}

\newcommand{\ket}[1]{\left| #1 \right\rangle}

\newcommand{\pg}[1]{\big\lfilet  #1  \big\rfilet}

\newcommand{\lpg}{\big\lfilet}
\newcommand{\rpg}{\big\rfilet}

\newcommand{\Nim}[1]{\textrm{Nim}(#1)}

\newcommand{\A}{Ruleset A}
\newcommand{\B}{Ruleset B}
\newcommand{\C}{Ruleset C}
\newcommand{\Cp}{Ruleset C'}
\newcommand{\D}{Ruleset D}

\newcommand{\NULL}{\textsc{null}}

\newcommand{\oL}{{\mathcal L}}
\newcommand{\oP}{{\mathcal P}}
\newcommand{\oR}{{\mathcal R}}
\newcommand{\oN}{{\mathcal N}}

\title{Toward Quantum Combinatorial Games}
\author{Paul Dorbec \institute{ Univ. Bordeaux, LaBRI, UMR 5800, F-33400 Talence, France \\CNRS, LaBRI, UMR 5800, F-33400 Talence, France}
\and
Mehdi Mhalla \institute{
 Univ. Grenoble Alpes, CNRS, Grenoble INP, LIG, F-38000 Grenoble France}
}
\def\titlerunning{Toward Quantum Combinatorial Games}
\def\authorrunning{P. Dorbec \& M. Mhalla}

\begin{document}
\maketitle

\begin{abstract}In this paper, we propose a Quantum variation of combinatorial games, generalizing the Quantum Tic-Tac-Toe proposed by Allan Goff [2006]. A combinatorial game is a two-player game with no chance and no hidden information, such as Go or Chess. In this paper, we consider the possibility of playing superpositions of moves in such games. We propose different rulesets depending on when superposed moves should be played, and prove that all these rulesets may lead similar games to different outcomes. We then consider Quantum variations of the game of Nim. We conclude with some discussion on the relative interest of the different rulesets. 
\end{abstract}

\section{Introduction}
Quantum information theory and its interpretations required the introduction of new  concepts such as superposition, entanglement, etc. To shed a new light on these concepts, it seems relevant to present them within the frame of a game. Moreover, introducing these phenomena in combinatorial game theory induces new families of games that had no natural motivation beforehand.

Quantum variations of ``economic games'' (with partial information) were already studied, see \cite{gavoille} for a 
discussion on that matter. 
However, to the best of our knowledge, the only combinatorial game with a tentative quantum variation is Tic-Tac-Toe. It was considered in \cite{gof-06}, ``as a metaphor for the counter intuitive nature of superposition exhibited by quantum systems'' and studied in \cite{LA10,solv2}.
In this paper, we want to propose a general way to provide a Quantum variation of any combinatorial game. 
Moreover, we want to improve the game interpretation of the measurement, in order to make it closer to what happens in quantum information theory.  
The main idea consists in allowing a player to play ``a superposition'' of moves.
Nevertheless, there are various ways to introduce Quantum variations of a combinatorial game in general, we consider some of them and argue on the interest of the different interpretations.

After defining some useful notions from classical combinatorial game theory, we propose a definition of Quantum variations of a game, with different rulesets. We then prove in Section~\ref{s:diff-rulesets} that each pair of rulesets differ for some game. Finally, we study the different rulesets on the game of Nim which is  a classical game for combinatorial game theory.

\section{Reminders on Classical Combinatorial Game Theory}

 A combinatorial game is a game with no chance and no hidden information, opposing two players, \emph{Left} and \emph{Right}. 
 In Combinatorial Game Theory (CGT), a game is described by a \emph{position} (e.g. the setting of the pieces on a board) and a \emph{ruleset}, that describes the legal \emph{moves} for each position and the positions to which these moves lead.
 We here consider short games, that is games with no possible infinite run and a finite number of moves. A classical example of such games is the game of Nim, that we consider often in the following. It is a two player game played on heaps of tokens. Each player alternately take any (positive) number of tokens from a single heap. When a player is unable to move (i.e. no heaps have any token left), he loses.

 To have a synthetic description of a game, a move is generally described in a general setting, independent of the current position (e.g. a bishop moves along a diagonal in Chess). In the following, we need a description of a move that can be interpreted with no ambiguity in different positions of the game. 
 In a formal way, we describe a combinatorial game with a set \(\cal{G}\) of positions and an alphabet \(\Sigma\) of all possible moves. Then a Ruleset is defined as a mapping :
\[\Gamma : \mathcal{G}\times \Sigma \rightarrow \mathcal{G}\cup\{\NULL\}\]
where \NULL\ means that the move is not legal from that position in the game. 

For example, in the game of Nim, we denote \(\Nim{x_1,\ldots,x_k}\) the position with $k$ heaps of $x_1,\ldots,x_k$ tokens. If the heaps contain at most \(n\) tokens, the alphabet of moves can be defined by pairs of integers: \(\Sigma=\{1,\ldots,k\}\times\{-n,\ldots,-1\}\), a pair \((i,-j)\) meaning we try to remove \(j\) tokens from heap number \(i\).
(In Partisan games where players have different moves, the moves \(m\in \Sigma\) are differentiated by their player.)

Under normal convention (that we follow here), any player that is unable to move from a position loses. 
A combinatorial game (opposing two players, \emph{Left} and \emph{Right}) may have four possible \emph{outcomes}: \(\oN\) when the first player can force a win, \(\oP\) when the second player can force a win, \(\oL\) when Left can force a win no matter who moves first and \(\oR\) when Right can force a win no matter who moves first. See \cite{Siegel} for more details.

Another fundamental concept in classical CGT is the sum of two games.
When playing on a sum, each player plays a move in one of the summand, until no moves is legal in any game. With our previous setting, from two rulesets 
\[\Gamma_1 : \mathcal{G}_1\times \Sigma_1 \rightarrow \mathcal{G}_1\cup\{\NULL\}\]
\[\Gamma_2 : \mathcal{G}_2\times \Sigma_2 \rightarrow \mathcal{G}_2\cup\{\NULL\}\]
we define the ruleset $\Gamma_{1+2}$  of the sum (assuming $\Sigma_1 \cap \Sigma_2 = \emptyset$) by
\begin{align*}
\Gamma_{1+2} : \qquad & \mathcal{G}_1\times \mathcal{G}_2\times (\Sigma_1\cup \Sigma_2) \rightarrow \mathcal{G}_1\times \mathcal{G}_2\cup\{\NULL\}\\
& (G_1,G_2,m) \rightarrow  \begin{cases}(\Gamma_1(G_1,m),G_2) \mbox{ if } m \in \Sigma_1 \mbox{ and } \Gamma_1(G_1,m) \neq \NULL\\
(G_1,\Gamma_2(G_2,m)) \mbox{ if } m \in \Sigma_2 \mbox{ and } \Gamma_2(G_2,m) \neq \NULL\\
\NULL \mbox{ otherwise.}\end{cases}
\end{align*}

Recall that two games \(G_1\) and \(G_2\) are said to be equivalent, denoted \(G_1\equiv G_2\), if for any third game \(G_3\),  \(G_1 + G_3\) as the same outcome as \(G_2 + G_3\). 
Observe that if all options of a game $\Gamma_1(G_1,m)$ are equivalent to the corresponding option of $\Gamma_2(G_2,m)$, the two games are equivalent.
The \emph{value} of a game is its equivalence class.
In particular, $*k$ denotes the value of a single Nim heap of $k$ tokens (we use notation with $*0=0$ and $*1=*$). 
In impartial games under normal convention, all games are equivalent to a Nim heap. Then the value of a game can be computed from the values of its options by the \emph{mex}-rule, i.e. the value of an impartial game is $*x$ where $x$ is the smallest non-negative integer such that $*x$ is not among the options.

The \emph{birthday} of a game is the maximum number of moves that can be played on the game. 

\section{Definition of a Quantum game}

 \medskip
 
\begin{figure}[h]
\begin{center}
\includegraphics[width=0.9\textwidth]{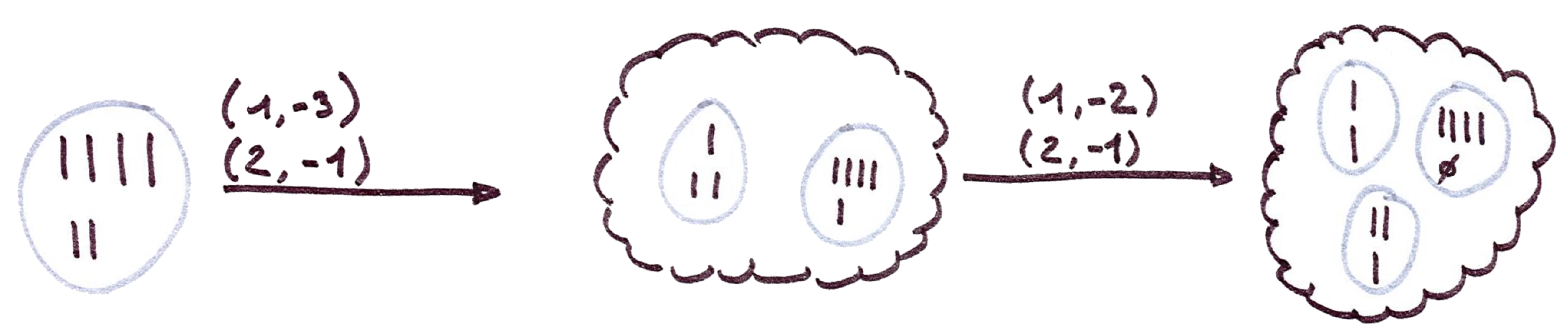}
\caption{Two superposed moves played on a game of Nim.}\label{f:clouds}
\end{center}
\end{figure}

From a classical game, we define a Quantum variation of the game (that can be seen as a modified ruleset) where players have to play Q-moves. A Q-move consists in a superposition (seen as a set) of classical moves. If the Q-move is reduced to a single classical move, we call it an {\em unsuperposed move}. As soon as a superposition of moves is played, the game is said in {\em Quantum state}. 
For each sequence of Q-moves, a {\em run} is a choice of a classical move among each superposed move of the sequence. If the choice forms a legal sequence of moves in the corresponding classical game, then the position obtained after this run is called a {\em realisation} of the game. For example, in Figure~\ref{f:clouds}, two superposed moves are played on a game of Nim, leading to only three realisations, since the sequence (1,-3)--(1,-2) is not legal.

In a Quantum state, a classical move is considered legal if there exists at least one realisation of the game where this move is legal (in the classical sense). A superposition of moves is legal if each of its superposed classical moves are legal (possibly in different realisations).
Observe that applying a move makes disappear all realisations inconsistent with that move.
The game ends when the next player has no legal moves. 
\medskip

 It should be noted here that with this setting, the way the Ruleset is defined influences the behaviour of its Quantum variation. This is why we cannot use simply the set of options of a position to describe a game, as it is in classical CGT. Moves need to be labeled and to have interpretations in different positions. 
 In particular, if we choose to describe the possible moves by the resulting position, then the game would become much simpler, since the positions obtained by any legal run would depend only of the last move.


\subsection{Quantum Rulesets definition}

Depending on when unsuperposed moves are considered legal, we here consider five rulesets.
\begin{description}
\item{\A}: Only superpositions of at least two (classical) moves are allowed.
\item{\B}: Only superpositions of at least two moves are allowed, with a single exception: when the player has only one possible move within all realisations together, he can play it as an unsuperposed move.
\item{\C}: Unsuperposed moves are allowed if and only if they are valid in all possible realisations.
\item{\Cp}: a player may play an unsuperposed moves if and only if it is valid in all possible realisations 
where he still has at least one classical move.
\item{\D}: Unsuperposed moves are always allowed (seen as the superposition of two identical moves).
\end{description}

Note that \Cp\ is the least permissive ruleset that is more permissive than \B\ and \C. We get the following lattice of permissivity:
\begin{center}
\begin{tikzpicture}[thick]
\draw (0,0) coordinate (D) -- ++(0,-1) coordinate (Cp) -- ++(-120:1) coordinate (B) -- ++(-60:1) coordinate (A) -- ++(60:1) coordinate (C) -- (Cp) ;
\draw (A) node[circle, fill=white] {A};
\draw (B) node[circle, fill=white] {B};
\draw (C) node[circle, fill=white] {C};
\draw (Cp) node[circle, fill=white] {C'};
\draw (D) node[circle, fill=white] {D};
\end{tikzpicture}
\end{center}

Remark that in all the above rulesets, the availability of a move depends only on the existence of a realisation, independently on how the realisations are correlated. Therefore, the history may be ignored, implying the following observation: 

\begin{observation}
 Independently of the sequence of moves, if two games are superpositions of the same realisations, the possible sequences of legal moves are identical. (In CGT, such games are said to be equal.)
\end{observation}

From this observation, we deduce that a position of a Quantum game is fully described by the possible realisations of the earlier moves, which form a multiset of classical positions. Since the multiplicity of a classical position is not important in all the above rulesets, we describe such a position as a set of superposed classical positions, denoted  \( \pg{G_1,\ldots,G_n} \). For example the superposition of positions \(G\) and \(G'\) is denoted \( \pg{G,G'} \), and the unsuperposed position \(G\) as the set \(\pg{G}\).
When we discuss on a specific Ruleset, e.g. \A, we use the letter as a subscript to the notation: $\pg{G_1,\ldots, G_n}_A$.

Legal classical moves in a superposition \( \pg{G_1,G_2,\ldots,G_n}\) now are any \(m\in \Sigma\) such that at least one of \(\Gamma(G_i, m)\) is not \NULL. In that case, we get the Quantized ruleset defined on $\mathcal{P}(\mathcal{G})\times \mathcal{P}(\Sigma)$ by 
\[\pg{G_1,\ldots,G_n}\times \{m_1,\ldots,m_k\}  \rightarrow \pg{\, \Gamma(G_i,m_j) \mid 1\le i\le n, 1\le j\le k, \Gamma(G_i,m_j) \neq \NULL\,}
\]

\subsection{Sum of Quantum games and Quantum sum of games}

Note that given this setting, a Quantum combinatorial game can be seen as a combinatorial game with a particular ruleset. We can thus sum a Quantum combinatorial game with another Quantum combinatorial game, or with a classical game, following the classical game sum (each player plays a move in one of the summand, until no moves is legal in any game). In such a sum of games, no superposition of moves distributed among the terms are allowed.


However, we can also consider the Quantum variation of the sum of two games, thus allowing superposition of moves distributed among the operands of the sum. Then we get \[\pg{G}+\pg{H} \neq \pg{G+H}\]

Observe that the Nim product $\otimes$, as defined in \cite{ONAG,Siegel}, allows to play a move on one operand or on both. In \D, where allowed moves are superpositions of any number of classical moves, we then obtain  \[\pg{G+H}_D =\pg{G}_D\otimes\pg{H}_D\,.\] 
This is not true for other Rulesets because in \(\pg{G+H}\), you could play only one element of the superposition in $G$ while it might not be a legal move in \(\pg{G}\).

\subsection{To play the game :}

If we want to play a Quantum variation of a game without computing all the realizations of the game at each move, it becomes difficult to know whether a move is legal or not. We suggest that each player at his turn announces the moves he wants to play. If his opponent suspects this move is not legal, he can challenge the player. Then the player must prove his move was legal: he must exhibit a run for each of the classical moves he played. If he manages to do so, then he wins, otherwise he loses. 
Note anyway that this does not change the outcome of the game: if a player can ensure a win in any of these settings, he can ensure a win in the other too.

\subsection{Link with Quantum theory}

In recent works, 
Abramsky and Bandenburger\cite{abrambandenb} gave a mathematical definition of contextuality, that modelize a more general setting than what was obtained with empirical models.
Then Acin et al.\cite{acin-al} showed that these phenomena can be exhibited in a pure combinatorial setting, representing contextuality scenarii with hypergraphs.
Somehow, our proposal of Quantum combinatorial games join in with this theoretical approach of contextuality.


Quantum combinatorial games have the following interpretation in quantum theory : the players act on a quantum system. Each move of a player consists in applying a unitary operator to the system,  
similarly to placing a gate chosen from a fixed set in a quantum circuit (as in the circuit model, see \cite{nielsen}).
If at some point, the move of some player brings the state of the system into a losing subspace, the player loses.

To decide if the state is in the losing subspace, one may apply a projective measurement. This produces a probabilistic game where each player can at best ensure a probability of winning. To recover the combinatorial version, we need to consider winning in a \emph{possibilistic} way : a player loses at the first step for which the projective measurement puts the state in the losing subspace with probability 1. This could be interpreted as allowing the player to replay the sequence of unitaries and the measure until he manages to prove his system is not in the losing subspace.

\smallskip

Note that when a game is in a quantum state, the possible classical states in the superposition can interact with each other and  a superposition of moves can act on different elements. Therefore, like in quantum theory, being in a superposition of positions is different than being in a probability distribution over states. 

As an illustration, it is possible to reproduce with our games the following behavior. 
From the basis states $\ket{0}$ and $\ket{1}$, applying the unitary Hadamard operator $H$ leads to the superpositions $\ket{0}+\ket{1}$ and $\ket{0}-\ket{1}$ respectively (We omit the normalization coefficient). These two states behave exactly like a random bit under measure. However, applying a second time the Hadamard operator $H$, we recover the initial quantum state. 



\sloppy
In the Quantum version of Nim, suppose that a player plays the superposed move $((1,-1),(2,-1))$ on $\Nim{1,1,2}$. It produces the superposition $\lpg \Nim{0,1,2},\Nim{1,0,2}\rpg$. If the second player plays also the same superposed move $((1,-1),(2,-1))$, it 
results in an unsuperposed game $\pg{\Nim{0,0,2}}$ and becomes a deterministic composition of the two moves. This way of obtaining a deterministic move out of two superposed moves reminds much of the previous behavior of the Hadamard operator.
In addition, the move $((1,-1),(2,-1))$ is not a legal move in any of the games  $\pg{\Nim{1,0,2}}_A$ and $\pg{\Nim{0,1,2}}_A$. 
So if one wants to see this definition of a quantum game as a probabilistic game, he needs to consider all the states in the probability distribution to determine the legal moves. Therefore, a probabilistic definition of such a game would be very artificial. Here is an other way of observing this problem: playing either on $\pg{\Nim{1,0,2}}_A$ or on $\pg{\Nim{0,1,2}}_A$, every move produces an element of the superposition that has maximal heapsize one. Thus a property that is true when playing on each of the elements of a superposition may be false when playing on  the superposition.

%
%




\section{Difference of rulesets} \label{s:diff-rulesets}

We now give some examples showing that all the rulesets are non-equivalent for some games.

\paragraph{Example 1}
The first example is based on the octal game $0.6$. This is a game played on a line of pins. Each player at his turn must remove a single pin that is not isolated (i.e. it still has at least one adjacent pin).
We number the pins from 1 to $n$ according to their position on the line. The octal game rule specifies that the move $i$ becomes illegal when both $\{i-1,i+1\}$ have been played.

We consider game \(0.6\) played on a line of four pins.
\begin{itemize}
\item the game has outcome \(\oP\) in the classical game.
\item it has outcome \(\oN\) in all other rulesets. A winning first move is \((1,4)\). Then whatever the answer to this move is, he can win playing $2$, $3$ or $(2,3)$.
\end{itemize}

\paragraph{Example 2}
Our second example is played on the game Domineering. This game is played on a subset of squares from a grid.
Left plays vertical dominoes (two adjacent squares) and right plays horizontal dominoes. A move is described by the squares the played domino occupies.

We play Domineering on the following board:

\begin{center}
\begin{tikzpicture}[scale=0.4]
\draw (3,1) -- (0,1)-- (0,0) -- (3,0) --(3,2) --(2,2) --(2,0);
\draw (1,1) -- (1,0);
\end{tikzpicture}
\end{center}

\begin{itemize}
\item The game has outcome \(\oR\) in the classical game. Indeed, if Left plays first, Right can still play a horizontal domino, while Right can prevent Left from being able to play with his first move. Right can use the same strategy in \D. Right is the only one who can play a superposition of moves, that makes him win in \A. Moreover, playing his winning move in the classical game, he wins in \C\ and \Cp. If Left starts in \C\ or \Cp, he must play an unsuperposed move and Right can answer also with an unsuperposed move.
\item In \B\, Right playing first is forced to play the superposed move, to which Left can answer. Left still loses when playing first.
\end{itemize}

\paragraph{Example 3} 
We consider the game called Hackenbush. We define here a restricted version that suffice to illustrate our example.
The game is played on a set of vertical paths composed of blue and red\footnote{We draw red edges with parallel lines.}  edges. At his turn Left selects a blue edge and remove it together with all the edges above it. Right does the same with a red edge. The game ends when a player has no edge to select. We label each edge of the game, and define a move by the label of the edge selected.

We play Hackenbush on the following game

\begin{center}
\begin{tikzpicture}[scale=0.6]
	\draw[color=red] (-0.025,0) -- (-0.025,1)  (0.975,1) -- (0.975,2) (1.975,1) --(1.975,2) (0.025,0) -- (0.025,1)  (1.025,1) -- (1.025,2) (2.025,1) --(2.025,2);
	\draw[color=blue, very thick] (1,0) -- (1,1) (2,0) -- (2,1);
	\draw (0.3,0.5) node[color=red] {$1$};
	\draw (1.3,1.5) node[color=red] {$2$};
	\draw (2.3,1.5) node[color=red] {$3$};
	\draw (1.3,0.5) node[color=blue] {$a$};
	\draw (2.3,0.5) node[color=blue] {$b$};
\end{tikzpicture}
\end{center}

In that game, in \A, Right cannot play after Left played twice ${\color{blue} (a,b)}$ and this is the only ruleset when Left wins playing first. The sequence ${\color{red} (2,3)}$-${\color{blue} (a)}$-${\color{red} (1,3)}$-${\color{blue} (b)}$ allows an extra move ${\color{red} (1)}$ only in \Cp\ and \D.

\begin{itemize}
\item the game has outcome \(\oP\) in the classical game as well as in \C.
\item it has outcome \(\oL\) in \A.
\item it has outcome \(\oR\) in \B, \Cp\ and \D.
\end{itemize}

\paragraph{Example 4}
We play Nim on two heaps of two tokens, i.e. $\Nim{2,2}$. This is $\oN$ in \D\ but $\oP$ in \Cp.
The key observation is that the  sequence $((1,-1),(2,-1))$ -- $((1,-2),(2,-1))$ -- $(1,-2)$ is winning for the first player in \D\ but the last move is not allowed in \Cp.

With these example, we deduce that all the above rulesets are pairwise different, i.e. for each pair of rulesets, there exists a game whose outcomes differ. Note that it is not necessary to use partisan games to exhibit the difference of the rulesets, but the above games were convenient to keep small examples.

\section{Results}

\begin{definition}
A game $G_1$ in a superposition \(\pg{G_1, G_2,\ldots,G_n}\) is said to be \emph{covered} by \(G_2,\ldots,G_n\) if either $G_1$ has no legal (classical) moves, or for every move $m \in \Sigma$ that is legal on $G_1$, there exist a game \(G_i, 2\le i\le n\) such that $m$ is legal on $G_i$ and \(\Gamma_1(G_1,m)\) is covered by \(\Gamma_2(G_2,m),\ldots, \Gamma_n(G_n,m)\).
\end{definition}

Observe for example that a Nim heap with $k$ tokens is covered by a Nim heap with more than $k$ tokens, which can easily be proved inductively.

\begin{theorem}\label{th:sup-simplification}
In \A, \B\ and \D, if a superposition of games \(\pg{G_1, G_2,\ldots,G_n}\) is such that  
$G_1$ is covered by \(G_2,\ldots,G_n\), then 
 \(\pg{G_1, G_2,\ldots,G_n}\equiv  \pg{G_2,\ldots,G_n}\).
\end{theorem}

\begin{proof}
We prove the result by induction on the birthday of $G_1$. First observe that if $G_1$ has no possible moves, then the equivalence is clear in \A, \B\ and \D\ (but not in \C\ and \Cp).
Now suppose the result is true for every game with birthday less than $G_1$'s birthday. Consider a  move \((m_1,\ldots,m_k)_{k\ge1}\). Then the game obtained is \(\pg{\Gamma(G_i,m_j),1\le i\le n, 1\le j\le k, \Gamma(G_i,m_j)\neq \NULL}\). For each $1\le j \le k$ such that \(\Gamma(G_1,m_j)\neq \NULL\), \(\Gamma(G_1,m_j)\) is covered by \(\Gamma_2(G_2,m_j),\ldots, \Gamma_n(G_n,m_j)\), so it can be removed from the superposition. Thus, every option of \(\pg{G_1, G_2,\ldots,G_n}\) is equivalent to the corresponding option of \(\pg{ G_2,\ldots,G_n}\) and we have equivalence.
\end{proof}

From our above observation on the Nim heaps, we get:

\begin{corollary}
In Rulesets A,B,D, the value of a superposition of a single Nim heap in different states depends only on the largest heap in the superposition.
\end{corollary}

\begin{corollary}\label{c:superposed-two}
In  \A, \B\ and \D, the Nim game on at most $k$ heaps is equivalent to the corresponding game where only superposition of at most $k$ moves are allowed.
\end{corollary}

\begin{proof}
We prove the result by induction on the birthday of the game: Consider 
a superposition of Nim heaps on at most $k$ heaps. Consider a superposed move from this game, say $((h,-i_h)_{i_h\in I_h,1\le h\le k}$ (i.e. the superposed removal of all numbers of token in $I_h$ for each heap $h$).
Theorem~\ref{th:sup-simplification} shows thatthis move is equivalent to the move $((h,-\min(I_h))_{1\le h\le k})$, which is a superposition of at most $k$ moves. In the case when only one $h$ is such that $I_h$ is non empty (and unsuperposed moves are not legal), we also get that it is equivalent to the move  
$((h,-\min(I_h),(h,-\min(I_h)-1))$.
\end{proof}

We remark that this corollary is not true in the case when the number of heaps is larger than $k$. For example, let $A_{|2}$ denote \A\ where only superpositions of two moves are allowed. 
We have that all the options of $\pg{\Nim{1,1,1}}_{A_{|2}}$ are equivalent to $\pg{\Nim{1,0,1},\Nim{0,1,1}}_{A_{|2}}$ (by permutation of the heaps). This position has two options, namely $\pg{\Nim{0,0,1}}_{A_{|2}}$ (by the move $(1,-1),(2,-1))$) which has value 0
and $\pg{\Nim{1,0,0},\Nim{0,1,0}, \Nim{0,0,1}}_{A_{|2}}$ (by any other move) which has value $*$. Thus $\pg{\Nim{0,1,1},\Nim{1,0,1}}_{A_{|2}}\equiv *2$ and $\pg{\Nim{1,1,1}}_{A_{|2}}\equiv 0$.

However playing the only possible superposition of three moves on $\pg{\Nim{1,1,1}}_A$ leads to the position $\pg{\Nim{0,1,1},\Nim{1,0,1},\Nim{1,1,0}}_{A}$. All options of this position are equivalent to 
$\pg{\Nim{0,0,1},\Nim{0,1,0},\Nim{1,0,0}}_{A}$ which has value $*$. As a consequence, we obtain that 
$\pg{\Nim{0,1,1},\Nim{1,0,1},\Nim{1,1,0}}_{A}\equiv 0$ and finally $\pg{\Nim{1,1,1}}_A \equiv *$.

\medskip

In the rest of this section, we get interested in the values of Quantum single heaps of Nim under the different rulesets.

\subsection{Values of single Nim heaps}

\begin{lemma}\label{l:A-OneHeap}
In \A, the value of a superposition of Nim heaps has value \(*(k-1)\) if and only if its largest heap has size \(k\ge 1\), i.e.
\[\pg{Nim(i_1),\ldots, Nim(i_\ell)}_A \equiv *(k-1)\mbox{ where }k=\max_{1\le j\le\ell}(i_j)\]
\end{lemma}

\begin{proof}
We prove the result by induction. If the largest heap is of size 1, then the game is \(\pg{\Nim{1}}\) or \(\pg{\Nim{0},\Nim{1}}\). No superposed move are possible, thus the game has value 0.

Assume now the lemma is verified up to \(k\). Then we claim a position with as largest heap \(k+1\) has value \(*k\). First, all the legal moves are of the form \((1,-x_i)\) for some \(1\le x_i \le k+1\). Denote $x=min(x_i)$. The move brings to a position whose largest heap is of size \(k-x+1\), which have value \(*(k-x)\) by induction. This proves the Lemma.
\end{proof}


\begin{lemma}
In \B, the value of a superposition of Nim heaps of size $k$ and less has value :
\[\pg{Nim(i_1),\ldots, Nim(i_\ell)}_B \equiv
\begin{cases}
0 \mbox{ if } k =2\\
* \mbox{ if } k =1\\
*(k-1) \mbox{ otherwise}
\end{cases} \mbox{ where }k=\max_{1\le j\le\ell}(i_j)\]
\end{lemma}

\begin{proof}
We prove the result by induction. If the largest heap is of size 1, then the game is \(\pg{\Nim{1}}\) or \(\pg{\Nim{0},\Nim{1}}\). There is a single legal move, which is the unsuperposed move $(1,-1)$ that lead to \(\pg{\Nim{0}}\). Thus the game as value \(*\).
If the largest heap is of size 2, then only the superposed move \(((1,-1),(1,-2))\) is legal, which brings to the game \(\pg{\Nim{0},\Nim{1}}\) of value \(*\). Thus it has value 0.

Assume now the lemma is verified up to \(k\ge 2\). Then we claim a position whose largest heap has size \(k+1\) has value \(*k\). First, all the legal moves are superpositions of the form \((1,-i)_{i\in I}\) with $I\subset [1,k+1]$.  
The move brings to a position whose largest heap is of size \(k-j+1\) where $j=\min \{i\in I\}$, which has value \(*(k-i)\) if \(k-i\ge 2\), $0$ if \(k-i= 1\) and $*$ if \(k-i= 0\), by induction. This proves the Lemma.
\end{proof}


\begin{lemma}

In \D, the value of a superposition of Nim heaps has value \(*k\) if and only if its largest heap has size \(k\), i.e.
\[\pg{Nim(i_1),\ldots, Nim(i_\ell)}_D \equiv *k\mbox{ where }k=\max_{1\le j\le\ell}(i_j)\]
\end{lemma}

\begin{proof}
Again the proof works by induction. In \D, a superposition with largest heap of size \(k\) has value at least \(*k\) since all unsuperposed moves are legal, and all its options have value at most \(*(k-1)\) by induction. 
\end{proof}

In contrast with what happens in the previous rulesets, the value of a superposition of Nim Heaps in \C\ does not depend only on the size of the maximum heap, as shows the following lemma.

\begin{lemma}
\label{l-C}
In \C, a superposition of Nim heaps with maximum size $k$ has value \(*(k-1)\), unless there is only one element in the superposition in which case the game has value \(*k\). i.e.
\[\pg{Nim(i_1),\ldots, Nim(i_\ell)}_C \equiv
\begin{cases}
*k \mbox{ if } \ell =1\\
*(k-1) \mbox{ otherwise}
\end{cases} \mbox{ where }k=\max_{1\le j\le\ell}(i_j)\]
\end{lemma}

\begin{proof}
Again the proof is by induction on \(k=\max_{1\le j\le\ell}(i_j)\).
If $\Nim{0}$ belong to the superposition (i.e. \(min_{1\le j\le \ell}(i_j)=0\)), no unsuperposed move is legal, then the options are all superposed. 
Also if the game is superposed with something else, all moves bring to a superposed game. In both cases the options are superposition of at least two Nim heaps of any maximum size $k'$ with  $0\le k' \le k-1$, which have value \(*(k'-1)\), so the game has value $*(k-1)$. 

Now if there is only one element in the superposition, all unsuperposed moves are legal, and the values of the options span the interval \([0,k-1]\).
\end{proof}


\begin{lemma} 
\label{l-Cp}
In \Cp, a superposition of Nim heaps with maximum size $k$ 
has value $*k$ unless $\Nim{k-1}$ belongs to the superposition. In that case, the game has value $*(k -1)$, with the only exceptions of $\pg{\Nim{0},\Nim{1}}_{C'}$ which has value $*$ and $\pg{\Nim{1},\Nim{2}}_{C'}$ which has value 0.


\[G=\pg{Nim(i_1),\ldots, Nim(i_\ell)}_{C'} \equiv
\begin{cases}
0  \mbox{ if } G = \pg{\Nim{1},\Nim{2}}_{C'} \\  \  \mbox{ or } \pg{\Nim{0},\Nim{1},\Nim{2}}_{C'} \\  
* \mbox{ if } k= 1 \\
*(k-1)  \mbox{ if } k-1 \in \{i_j\}\\
*k  \mbox{ otherwise}
\end{cases} \mbox{ where }k=\max_{1\le j\le\ell}(i_j)\]
\end{lemma}

\begin{proof}
We first consider separately the games with a largest heap of size less than $3$.
Observe first that in \Cp, the games \(\pg{\Nim{1}}\) and \(\pg{\Nim{0},\Nim{1}}\) have only $(1,-1)$ as a possible move, which leads to an ended game. Thus they have value $*$. Now in both \(\pg{\Nim{1},\Nim{2}}\) and  \(\pg{\Nim{0},\Nim{1},\Nim{2}}\), the options are \((1,-1)\) and \(((1,-1),(1,-2))\). Both lead to the game \(\pg{\Nim{0},\Nim{1}}\) which has value $*$, so the initial game has value $0$. 
Finally, the games $\pg{\Nim{2}}$ and $\pg{\Nim{0},\Nim{2}}$ have options \(\pg{\Nim{0}}\), \( \pg{\Nim{1}}\), and \(\pg{\Nim{0,1}}\) and thus have value $*2$.

Consider now the other games, which have a largest heap in the superposition of size $k\ge 3$. We prove the lemma by induction on $k$. Suppose first that $\Nim{k-1}$ belongs to the superposition. Then the move \(((1,-j),(1,-j-1))\) leads to a superposition of largest heap $\Nim{k-j}$ that also contains $\Nim{k-j-1}$, thus of value $*(k-j-1)$ if $1\le j\le k-3$, $0$ if $j=k-2$, and $*$ if $j=k-1$. Moreover, all moves bring to a game with largest heap of size $k'$ with $k'\le k-1$ and that contain $\Nim{k'-1}$. 

Assume now that $\Nim{k-1}$ does not belong to the superposition. Then the move \(((1,-j),(1,-k))\) (for $1\le j \le k-2$) leads to a superposition of largest heap $\Nim{k-j}$ that does not contain $\Nim{k-j-1}$, so of value $*(k-j)$ by induction. 
With \(((1,-k),(1,-k+1))\), we reach the game \(\pg{\Nim{0},\Nim{1}}\) which has value $*$.
With the move  \(((1,-k),(1,-k+1),(1,-k+2))\), we reach the game \(\pg{\Nim{0},\Nim{1},\Nim{2}}\) which has value $0$. All moves reach a game with value at most $*(k-1)$ since the largest heap of the resulting game is of size at most $k-1$. This concludes the proof.
\end{proof}





\section{Conclusion}

\subsection{Quantum variations of games with more than one heap.}

In the above, we restricted ourselves to the study of games with a single Nim heap, as an initial approach. The next step is naturally to consider games where there are two heaps or more. Such a game is fundamentally different from the sum of two games on one heap, since it becomes possible to play a superposition of moves on different heaps. We then have to consider some extra cases like in the following lemma.

\begin{lemma}
In \A, for $i,j>0$,
 \(\pg{\Nim{i,0},\Nim{0,j}}_A \equiv *(\max(i,j) -1+\delta_{i,j})\) where $\delta_{i,j}$ denotes the Kronecker delta.
\end{lemma}

\begin{proof}
Let $S=\pg{\Nim{i,0},\Nim{0,j}}$, with  $i\ge 0,\, j> 0$. Observe first that if \(i=0, \, j=1\), then there is no possible  move and the game has value 0.

We now consider the possible moves on $S$. First observe that every possible move brings to a game that is equivalent to a superposition of at most two games. Indeed, a superposition of classical moves on the first heap 
\(((1,i_1),\ldots (1,-i_k))\) 
brings to a superposition \(\pg{\Nim{i-i_1,0},\ldots \Nim{i-i_k,0}}\), which is equivalent to
\(\pg{\Nim{i-\min_l \{i_l\},0}}\) 
by Theorem~\ref{th:sup-simplification}. This option has value $*(i-\min_l \{i_l\} -1)$ by Lemma \ref{l:A-OneHeap}, which is at most $*(i-2)$.
Similarly, moves \(((2,-j_1),\ldots (2,-j_k))\) bring to the position  \(\pg{\Nim{0,j-\min_l \{j_l\}}}\) with value at most $*(j-2)$. 
Finally, moves of the form \(((1,-i_1),\ldots (1,-i_l), (2,-j_1),\ldots (2-j_{k'})\) 
lead to the superposition \(\pg{\Nim{i-\min_l \{ i_l\},0},\Nim{0,j-\min_l \{j_l\}}}\).

By induction, these positions have value at most $*(\max\{i,j\}-1)$, and the only situation when this maximum value is attained is when $i=j$, and the move played is \(((1,-1),(2,-1))\). This prove that the value of \(\pg{\Nim{i,0},\Nim{0,j}}_A\) is at most \(*(\max(i,j) -1+\delta_{i,j})\). 

To show the other inequality, we simply need to observe that all the values $*k$ for $k\le j-2$ are reached by the moves $((2,-j+k),(2,-j+k+1))$, and similarly for $i$ on the first heap. 
\end{proof}

\subsection{Discussion on the different rulesets.}

Among the proposed rulesets, we think that \C\ is interesting for its physical interpretation: when the position is classical we can use classical moves but when applying a superposed move we put the system in a superposition and we can no-longer act classically: one cannot force a branch of a superposition by choosing a classical move.

On the other hand, the Combinatorial Game Theorist might find the \A\ more interesting to study, since it diverges most from the classical game. 

\smallskip

Another natural restriction for simplifying the game study would be to limit moves allowed to superpositions of at most two moves. This was the choice made by the Quantum Tic-Tac-Toe that can be played online (e.g. at \url{http://countergram.com/qtic/}) or by the \emph{Quantum tic tac toe}-apps for smartphones.  
By Corollary~\ref{c:superposed-two}, this restriction does not change anything in \A, B or D on Nim games on at most two heaps. Lemmas~\ref{l-C} and~\ref{l-Cp} (for \C\ and C') also hold under this restriction.


With this restriction, we computed with CGSuite \cite{CGSuite} 
the values of the games \(\pg{\Nim{i,j}}\) for small values of $i$ and $j$, they are given in the Appendix, for each of the rulesets. Observe in particular the case of the game $\pg{\Nim{3,3}}_A$ that has value 4, which cannot be obtained by a combination of sums of heaps on at most three tokens. 


\section*{Acknowledgments}
This research was supported in part by the ANR-14-CE25-0006 project of the French National Research Agency and by PEPS INS2I 2016 JCQ.

The authors would also like to thank Simon Perdrix and Urban Larsson for fruitful discussions.

\providecommand{\urlalt}[2]{\href{#1}{#2}}
\providecommand{\doi}[1]{doi:\urlalt{http://dx.doi.org/#1}{#1}}

\newpage

\section*{Appendix}




We first give the values of small games on two heaps computed with CGSuite in the case when superposition of only two moves are allowed. The tables are available for each ruleset.


\begin{table}[h]
\centering
\begin{tabular}{c|cccccccccccc}
  & 0 & 1 & 2 & 3 & 4 & 5 & 6 & 7 & 8 & 9 & 10&11\\
\hline
0 & 0 & 0 & 1 & 2 & 3 & 4 & 5 & 6 & 7 & 8 & 9 &10\\
1 & 0 & 0 & \green{3} & 1 & 2 & 5 & 4 & 7 & 6 & 9 & 8 & 11\\
2 & 1 & \green{3} & 1 & 0 & \green{6} & 2 & 7 & 4 & 9 & 10 &5& 8\\
3 & 2 & 1 & 0 & \red{4} & \green{7} & 3 & 2 & 8 &5&6 & \blue{13} & \blue{14}\\
4 & 3 & 2 & \green{6} & \green{7} & 0 & \green{1} & 3 & \blue{11} & \blue{12} & 4 & \blue{14} & 9 \\
5 & 4 & 5 & 2 & 3 & \green{1} & 0 & \blue{11} & \blue{12} & \blue{13} &\blue{14}&4&\blue{16}\\
6 & 5 & 4 & 7 & 2& 3&\blue{11}&\blue{12}&\blue{13}&0&\blue{15}\\ 
\end{tabular}
\smallskip
\begin{center} Values of \(\pg{\Nim{i,j}}_A\) \end {center}
\label{tab:A}
\end{table}

\begin{multicols}{2}

\centering
\begin{tabular}{c|ccccccccccc}
  & 0&1 & 2 & 3 & 4 & 5 & 6 &7&8 &9 &10\\
\hline
0 &  0 &1&0 &2&3&4&5&6&7&8&9\\
1 &1& 0 & 1 & 4 & 2 & 3 & 6 &5&8\\
2 &0&1& 0&5&6&2&3&8\\
3 &2&4&5&1&7&6&0\\
4 &3&2&6&7&3&4 \\
5 &4&3&2&6&4&0
\end{tabular}

\smallskip 

{Values of \(\pg{\Nim{i,j}}_B\)}

\bigskip


\centering
\begin{tabular}{c|ccccccccccc}
  & 0&1 & 2 & 3 & 4 & 5 & 6 &7 \\
\hline
0 &  0 &1&2&3&4&5&6&7\\
1 & 1&0&3&4&5&6&7&8\\
2 &2&3&1&5&6&7&8&9 \\
3 &3&4&5&0&7&8&2&10\\
4 &4&5&6&7&0&9&10 \\
5&5&6&7&8&9&0 \\
6&6&7&8&2&10\\
7&7&8&9&10
\end{tabular}

\smallskip

{Values of \(\pg{\Nim{i,j}}_C\)}

\columnbreak


\centering

\begin{tabular}{c|ccccccccccc}
  & 0&1 & 2 & 3 & 4 & 5 & 6 &7&8 \\
\hline
0 &  0 &1&2&3&4&5&6&7&8\\
1 & 1&0&3&4&2&6&7&5&9\\
2 &2&3&0&5&6&1& 8\\
3 &3&4&5&1&7&0\\
4 &4&2&6&7&3
\end{tabular}

\smallskip

{Values of \(\pg{\Nim{i,j}}_{C'}\)}

\bigskip


\centering

\begin{tabular}{c|ccccccccccc}
  & 0&1 & 2 & 3 & 4 & 5 & 6 &7 \\
\hline
0 &  0 &1&2&3&4&5&6&7\\
1 &1&0&3&2&5&4&7&6 \\
2 &2&3&1&0&6&7&4 \\
3 &3&2&0&4&1&6\\
4 &4&5&6&1&0&2 
\end{tabular}

\smallskip

{Values of \(\pg{\Nim{i,j}}_D\)}
\end{multicols}


\end{document}